\documentclass[11pt]{article}
\usepackage{amsfonts,amsmath,amssymb,amsthm,graphicx,url}
\usepackage{color}
\usepackage{algorithm}
\usepackage{algorithmic}
\usepackage{relsize}
\usepackage[margin=1in]{geometry}

\newtheorem{definition}{Definition}
\newtheorem{theorem}{Theorem}
\newtheorem{lemma}{Lemma}

\DeclareMathOperator*{\argmax}{arg\,max}

\newcommand{\bE}{\mathbb{E}}

\newcommand{\bI}{\mathbb{I}}

\newcommand{\bR}{\mathbb{R}}

\newcommand{\cI}{{\cal I}}

\newcommand{\cM}{{\cal M}}

\title{Approximately Maximizing the Broker's Profit in a Two-sided Market}

\author{Jing Chen$^{\dag}$ \hspace{30pt} Bo Li$^{\dag}$ \hspace{30pt} Yingkai Li$^{\ddag}$\\
$^{\dag}$Department of Computer Science, Stony Brook University, USA\\
\texttt{\{jingchen, boli2\}@cs.stonybrook.edu}\\
$^{\ddag}$Department of Computer Science, Northwestern University, USA\\
\texttt{yingkai.li@u.northwestern.edu}}

\begin{document}

\maketitle

\begin{abstract}
We study how to maximize the broker's (expected) profit in a two-sided market, 
where she buys items from a set of sellers
and resells them to a set of buyers.
Each seller has a single item to sell and holds a private value on her item,
and each buyer has a valuation function over the bundles of the sellers' items.
We consider the Bayesian setting where the agents' values are independently drawn from prior distributions,
and aim at designing dominant-strategy incentive-compatible (DSIC) mechanisms that are approximately optimal.
%As a bridge between one-sided markets (i.e. auctions) and two-sided markets,
%we also consider {\em production-cost markets} where each item has a cost for it to be produced and the costs are
%publicly known to everybody.

{\em Production-cost markets}, where each item has a publicly-known cost to be produced,
provide a platform for us to study two-sided markets.
Briefly, we show how to covert a mechanism for production-cost markets into a mechanism for the broker,
whenever the former satisfies {\it cost-monotonicity}.
This reduction holds even when buyers have general combinatorial valuation functions.
When the buyers' valuations are additive, we
generalize an existing mechanism to production-cost markets
in an approximation-preserving way.
We then show that the resulting mechanism is cost-monotone and thus can be converted into
an 8-approximation mechanism for two-sided markets.

\end{abstract}

\section{Introduction}\sloppy
Two-sided markets are widely studied markets in economics
\cite{myerson1983efficient,mcafee1992dominant,mcafee2008gains},
where a number of buyers and a number of sellers are connected by an intermediary, 
such as antique markets, used-car markets, and pre-owned house markets. 
Here each seller has a single item to trade for money 
and holds a private value for her owned item,
while each buyer's private information is a general combinatorial valuation function over the bundles of the sellers' items.
A common feature in these situations is that
the intermediary keeps the difference between the payments made by the buyers and the payments made to the sellers ---that is, the intermediary's {\em profit}.
We call such an intermediary a {\em broker}.
The objective of the broker is to acquire the items from the sellers 
and resell them to the buyers to maximize her profit. 
The problem studied in our paper is to design the mechanism 
in the two-sided market that maximize the broker's profit.
For convenience, we refer to the sub-market between the sellers and the broker the {\em seller-side market}
and to the sub-market between the broker and the buyers the {\em buyer-side market}.

If the broker had all the items, then we would only have the buyer-side market, which is an auction
where the broker tries to maximize her revenue. Auctions have been well studied in the literature following the seminal work of Myerson \cite{myerson1981optimal}. In Section \ref{sec:related}, we will briefly recall the most relevant literature on auctions.
If the broker would keep the items, then we only have the seller-side market, which is a procurement game.
Budget feasible procurement
has been studied by many in the Algorithmic Game Theory literature \cite{singer2010budget,dobzinski2011mechanisms,chen2011approximability,chan2014truthful}.
%where each seller has one item or multiple copies of the same item.
The broker wants to maximize her value for the items she buys,
subject to a budget constraint.

Although auctions and procurements are closely related to the broker's problem,
they cannot be dealt with separately in two-sided markets.
%since the items are not held by the broker and the broker does not face a fixed budget.
Indeed, the difficulty of the broker's problem is to simultaneously and truthfully
elicit both the sellers' and the buyers' valuations, so as to generate a good profit.

\subsection{Main Results and Techniques}
%To establish a bridge between auction (auction) and two-sided market,
In this paper we assume the values of the sellers and buyers are independently distributed,
and we study simple {\em dominant-strategy incentive compatible} (DSIC) mechanisms. %drawn from
%some prior distributions.
To approximately maximize the (expected) profit of the broker, we first develop a reduction,
through which we can directly convert mechanisms for {\em production-cost markets} into mechanisms for two-sided markets.
In a production-cost market, % as a bridge in the reductions,
% connecting these two scenarios,
the broker is able to produce all the items, each item has a cost to be produced and the costs are publicly known.
%to establish a bridge between auctions and two-sided markets.
%We first show that production-cost markets are reducible to auctions for arbitrary combinatorial valuations of the buyers.
Roughly speaking, we say a mechanism for production-cost markets is {\em cost-monotone}
if, when the cost of an item increases, the likelihood that it is sold does not increase.
We show that any cost-monotone mechanism for production-cost markets can be converted into
a mechanism for two-sided markets via a black-box approach.
This reduction holds for general combinatorial valuation functions of buyers.
%Based on this market, we show the following two reductions.
%which preserves the optimality or approximation ratio.

%\paragraph
%\noindent
%\sloppy
%{\em
%Theorem \ref{thm:reduction:1} (Informal).
% Any DSIC mechanism that is an $\alpha$-approximation for combinatorial auctions can be converted into
%a DSIC mechanism that is an $\alpha$-approximation for production-cost markets.
%}

%\paragraph
{\em Theorem \ref{thm:reduction:2} (Informal).
Any cost-monotone DSIC mechanism that is an $\alpha$-approximation for production-cost markets, 
%combinatorial cost-monotone mechanism,
can be converted into a DSIC mechanism that is an $\alpha$-approximation for two-sided markets.
}
%
%Roughly speaking, {\em cost-monotonicity}
%means that if an item's cost is increased,
%the likelihood of selling it should not be increased.

Next, we use cost-monotonicity as a guideline in constructing concrete mechanisms for two-sided markets.
When the buyers have additive valuations,
we generalize the duality framework of \cite{cai2016duality} and the mechanism there to design
a cost-monotone mechanism for production-cost markets.
%which is a generalization of an existing mechanism \cite{cai2016duality}.
Following our reduction, we immediately obtain a mechanism for two-sided markets.

\smallskip

{\em Theorem \ref{thm:add} (Informal).
When the buyers have additive valuations,
there exists a DSIC mechanism for two-sided markets which is an 8-approximation to the optimal profit.
}

\subsection{Related Work}\label{sec:related}

{\em Bayesian auctions} have been extensively studied since the seminal work of \cite{myerson1981optimal}.
% is a special market structure of two-sided markets,
%which is widely studied in the literature
%\cite{myerson1981optimal,cremer1988full,archer2001truthful,ronen2001approximating,alaei2014bayesian}.
For single-parameter settings, Myerson's mechanism %\cite{myerson1981optimal}
is optimal. The problem becomes more complicated in multi-parameter settings \cite{hart2017approximate}.
Although optimal Bayesian incentive-compatible (BIC) mechanisms have been characterized \cite{cai2012optimal,cai2012algorithmic},
they are too complex to be practical. Also, optimal DSIC mechanisms remain unknown.
Thus, simple DSIC mechanisms that are approximately optimal have been studied in the literature, 
such as \cite{kleinberg2012matroid,yao2015n,cai2016duality}

%such as when the buyers' valuations are unit-demand \cite{chawla2010multi,kleinberg2012matroid}, additive \cite{yao2015n,cai2016duality} or sub-additive~\cite{cai2017simple}.

{\em Two-sided markets} are also called double auctions  \cite{mcafee1992dominant},
bilateral trading \cite{myerson1983efficient} or market intermediation \cite{jain2012ebay} in the literature.
Maximizing the broker's profit is an important objective for two-sided market.
The seminal paper \cite{myerson1983efficient} characterized the optimal mechanism for one seller and one buyer,
which is further generalized by \cite{deng2014revenue}
to multiple single-parameter sellers and buyers.
Unlike our work, \cite{deng2014revenue} studies the Bayesian Incentive Compatible (BIC) mechanisms. 
DSIC mechanisms are also studied in the literature, but only for some special cases:
\cite{jain2012ebay} studies the case of a single buyer and multiple sellers,
\cite{balseiro2019dynamic} studies the case of a single seller and multiple buyers,
and \cite{gerstgrasser2016revenue} studies the optimal mechanism when the numbers of sellers and buyers are both constants.
Although \cite{chan2016budget} studies two-sided markets with multiple buyers and multiple sellers,
the dealer there has a fixed budget and their mechanism guarantees that the payment to sellers is within the budget.
Before our work, it remained unknown how to design a (simple) DSIC mechanism that approximates the optimal profit
in multi-parameter settings with a general number of sellers and buyers.

%It is shown in \cite{myerson1983efficient} that even with one seller and one buyer,
%no truthful mechanism can achieve maximum social welfare.

Finally, we briefly discuss the efficiency of two-sided markets, which is measured by {\em gain-from-trade} (GFT),
i.e.,  the total value gained by the buyers minus the value contributed by the sellers.
\cite{mcafee1992dominant} gave the first approximation mechanism 
for the one seller and one buyer case, 
%when the median of the buyer's distribution is higher than the seller's,
%there exists a 2-approximation mechanism to the optimal GFT.
and \cite{cai2017approximating} gives approximation mechanisms for multiple buyers with unit demand valuations.
Recently, \cite{segal2018double} and \cite{segal2018muda} study the asymptotically efficient mechanisms instead of constant approximations.
For maximizing {\em social welfare}, \cite{colini2016approximately,colini2017approximately} provide constant-approximation mechanisms.

\section{Preliminaries}

%\paragraph{Two-Sided Markets.}
A two-sided market includes a set $M$ of $m$ sellers,
and a set $N$ of $n$ buyers.
We consider the setting where each seller $j$ has one item $j$ to sell, so
%So, the item set is the same as the seller set
%and when there is no ambiguity,
we may refer to items and sellers interchangeably.
%The broker wants to maximize her profit
%through buying the items from the sellers and reselling them to the buyers.
The total payment made by the buyers is the broker's {\em revenue}, and her {\em profit} is the revenue
minus the total payment to the sellers.
%) is her {\em profit}.
%Throughout this paper we use $i\in  N$ to denote a buyer and $j\in  M$ to denote a seller or an item.
%We assume the buyers are males and the sellers are females.

Each buyer $i$ has valuation $v^{B}_{i}: 2^{ M} \to \bR^{+} \cup \{0\}$ with $v^{B}_{i}(\emptyset)=0$.
%mapping a set of items to a nonnegative number, with $v^{B}_{i}(\emptyset)=0$.
The function $v^{B}_{i}$ is monotone: for any $T\subseteq S\subseteq M$, $v^{B}_{i}(T)\leq v^{B}_{i}(S)$.
%Without loss of generality,.
In our reduction between production-cost and two-sided markets, we consider combinatorial valuations
and do not impose any restriction on $v^{B}_{i}$.
%Thus a buyer's values for can be arbitrari valuations can be combinatorial.
%\footnote{A valuation function $v$ should be free-disposal: for any $T\subseteq S$, $v(T)\leq v(S)$.}.

Each function $v^{B}_{i}$ is independently drawn from a distribution $D^{B}_{i}$ over the set of all possible valuation functions, with density function $f^{B}_{i}$ and cumulative probability $F^{B}_{i}$.
Let $D^{B} = \times_{i\in  N}D^{B}_{i}$, $f^{B} = \times_{i\in  N}f^{B}_{i}$ and $F^{B} = \times_{i\in  N}F^{B}_{i}$.
%$t^{B}_{ij}\in \bR^{+} \cup \{0\}$,
%is independently drawn from distribution $D^{B}_{ij}$ whose density function is $f^{B}_{ij}$.
%Denote by $t^{B}_{i}=(t^{B}_{ij})_{j\in  M}$ ($D^{B}_{i}=\times_{j\in  M} D^{B}_{ij}$, $f^{B}_{i}=\times_{j\in  M} f^{B}_{ij}$) and
%$t^{B}=(t^{B}_{i})_{i \in  N}$ ($D^{B}=\times_{i\in  N}D^{B}_{i}$, $f^{B}=\times_{i\in  N}f^{B}_{i}$) the product distributions.
%Each buyer $i$'s value on a set $S$ when her item value is $t^B_i$ is $v^B_i(S, t^B_i)$,
%where $v^B_i$ is a monotone subadditive function.
%The item value $t_i$ is the private information of buyer $i$ while the function $v_i$ is public information.
Each seller $j$'s value on her item, $v^{S}_{j}\in \bR^{+} \cup \{0\}$,
is independently drawn from a distribution $D^{S}_{j}$, with density function $f^{S}_{j}$ and cumulative probability $F^{S}_{j}$.
Let $D^{S} = \times_{j\in  M}D^{S}_{j}$, $f^{S} = \times_{j\in  M}f^{S}_{j}$ and $F^{S} = \times_{j\in  M}F^{S}_{j}$.
%We suppose all the distributions have finite support and we will generate our result to continuous cases later.
Let the supports of distributions $D^{B}_{i}$ and $D^{S}_{j}$ be $T^{B}_{i}$ and  $T^{S}_{j}$, respectively.
$T^{B}_{i}$ and  $T^{S}_{j}$ are called the {\em valuation spaces} of buyer $i$ and seller~$j$.
Let $T^{B}=\times_{i\in  N} T^{B}_{i}$ and $T^{S}=\times_{j\in  M} T^{S}_{j}$.
Finally, denote by $\cI = ( N, M,D^{B},D^{S})$ a two-sided market instance.

A mechanism $\cM$ for two-sided markets is %a tuple of four functions 
represented by $(x^{B},x^{S},p^{B},p^{S})$.
Given a valuation profile $(v^{B},v^{S})$,
%$\cM(v^{B},v^{S})=(x^{B}(v^{B},v^{S}),x^{S}(v^{B},v^{S}),p^{B}(v^{B},v^{S}),p^{S}(v^{B},v^{S}))$,
\begin{itemize}
\item $x^{B}(v^{B},v^{S})\triangleq(x^{B}_{i}(v^{B},v^{S}))_{i\in  N}$
is the allocation of the buyers, where $x^{B}_{i}(v^{B},v^{S})=(x^{B}_{iA}(v^{B},v^{S}))_{A \subseteq  M}$
with $x^{B}_{iA}(v^{B},v^{S})\in [0,1]$, representing the probability that buyer~$i$ gets the item set $A$,
%when the reported valuation profiles are
under valuation profile $v^{B}$ and $v^{S}$.
Moreover, $\sum_A x^{B}_{iA} (v^{B},v^{S})= 1$.
\item $x^{S}(v^{B},v^{S})=(x^{S}_{j}(v^{B},v^{S}))_{j\in  M}$ is the allocation of the sellers with $x^{S}_{j}(v^{B},v^{S})\in [0,1]$,
representing the probability that seller $j$'s item is sold under $(v^{B}, v^{S})$.

\item $p^{B}(v^{B},v^{S})=(p^{B}_{i}(v^{B},v^{S}))_{i \in  N}$ is the payment of the buyers,
% under $(v^{B}, v^{S}$),
where $p^{B}_{i}(v^{B},v^{S}) \in \bR^{+}\cup \{0\}$.

\item $p^{S}(v^{B},v^{S})=(p^{S}_{j}(v^{B},v^{S}))_{j \in  M}$ is the payment to the sellers,
%when the buyers' valuation profile is $v^{B}$ and the sellers' valuation profile is $v^{S}$,
where $p^{S}_{j}(v^{B},v^{S}) \in \bR^{+}\cup \{0\}$.
\end{itemize}

A {\em feasible} mechanism $\cM$
% = (x^{B},x^{S},p^{B},p^{S})$ satisfies
is such that
$$\sum_{A \ni j} \sum_{i\in  N}x^{B}_{iA}(v^{B},v^{S}) \leq x^{S}_{j}(v^{B},v^{S})$$
for any item $j\in  M$ and any valuation profile $(v^{B},v^{S})$.
In principle, the above condition may allow a mechanism to sell an item that it didn't buy or to buy an item without selling it. However, these cases never happen in the mechanisms
in this paper.\footnote{
Note that our feasibility constraint only requires ``feasible in expectation'' which is weaker than ex post feasibility. 
All of our results still hold if we change the requirement to be ex post feasible.}
The expected profit $PFT(\cM;\cI)$ of mechanism $\cM$ for instance $\cI$ is
$$\mathop{\bE}\limits_{v^{S}\sim D^{S}; v^{B}\sim D^{B}}
\sum_{i\in  N} p^{B}_{i}(v^{B},v^{S})  - \sum_{j\in  M} p^{S}_{j}(v^{B},v^{S}).$$

The utilities of the agents are quasi-linear. That is,
for each buyer~$i$, for any valuation subprofile $v^{B}_{-i}$ of the buyers and
any valuation profile $v^{S}$ of the sellers,
when $i$ reports her true valuation function $v^{B}_{i}$, her utility under mechanism $\cM$ is
$$u^{B}_{i}(v^{B}_{i};\cM, v^{B}_{-i},v^{S}) = \sum_{A\subseteq M} x^{B}_{iA}(v^{B},v^{S})  v^B_i(A) - p^{B}_{i}(v^{B},v^{S}).$$
%\begin{eqnarray*}
%&&u^{B}_{i}(v^{B}_{i};\cM, v^{B}_{-i},v^{S}) \\
%&=& \sum_{A\subseteq M} x^{B}_{iA}(v^{B},v^{S})  v^B_i(A) - p^{B}_{i}(v^{B},v^{S}).
%\end{eqnarray*}
For each seller $j$, for any valuation subprofile $v^{S}_{-j}$ and $v^{B}$, % any valuation profile 
when $j$ reports her true value $v^{S}_{j}$, her utility is
$$u^{S}_{j}(v^{S}_{j};\cM,v^{B},v^{S}_{-j}) = p^{S}_{j}(v^{B},v^{S}) - v^{S}_{j}x^{S}_{j}(v^{B},v^{S}).$$

Mechanism $\cM$ is {\em dominant-strategy incentive-compatible} (DSIC) if:
(1) for any buyer $i$, $v^{B}_{-i}$, $v^{S}$, and $v^{B}_{i}$, $v'^{B}_{i}$, %valuation functions 
\begin{eqnarray*}
u^{B}_{i}(v^{B}_{i};\cM,v^{B},v^{S}_{-i}) 
\geq \sum_{A\subseteq M} x^{B}_{iA}(v'^{B}_{i},v^{B}_{-i},v^{S}) v^B_i(A) - p^{B}_{i}(v'^{B}_{i},v^{B}_{-i},v^{S});
\end{eqnarray*}
and (2)  for any seller $j$, $v^{S}_{-j}$, $v^{B}$ and $v^{S}_{j}$,~$v'^{S}_{j}$, % valuation profiles , 
$$u^{S}_{j}(v^{S}_{j};\cM,v^{S}_{-j},v^{B}) \geq p^{S}_{j}(v^{B}, v'^{S}_{j},v^{S}_{-j}) - v^{S}_{j}x^{S}_{j}(v^{B}, v'^{S}_{j},v^{S}_{-j}).$$
%\begin{eqnarray*}
%&&u^{S}_{j}(v^{S}_{j};\cM,v^{S}_{-j},v^{B}) \\
%&\geq& p^{S}_{j}(v^{B}, v'^{S}_{j},v^{S}_{-j}) - v^{S}_{j}x^{S}_{j}(v^{B}, v'^{S}_{j},v^{S}_{-j}).
%\end{eqnarray*}

Mechanism $\cM$ is {\em individually rational} (IR) if:
(1) for any buyer $i$, $v^B_i$, $v^{B}_{-i}$ and $v^{S}$,
%and valuation function $v^{B}_{i}$,
$u^{B}_{i}(v^{B}_{i};\cM,v^{B}_{-i},v^{S}) \geq 0;$
%\sum_{A\subseteq M} x^{B}_{iA}(v'^{B}_{i},v^{B}_{-i},v^{S}) v^B_i(A) - p^{B}_{i}(v'^{B}_{i},v^{B}_{-i},v^{S});$$
and (2)  for any seller $j$, $v^S_j$, $v^{S}_{-j}$ and $v^{B}$,
%and values $v^{S}_{j}$,
$u^{S}_{j}(v^{S}_{j};\cM,v^{S}_{-j},v^{B}) \geq 0.$
%p^{S}_{j}(v^{B}, v'^{S}_{j},v^{S}_{-j}) - v^{S}_{j}x^{S}_{j}(v^{B}, v'^{S}_{j},v^{S}_{-j});$$

Mechanism $\cM$ is {\em Bayesian incentive-compatible} (BIC) if~(1) for any buyer $i$ and valuation functions $v^{B}_{i}$, $v'^{B}_{i}$,
%\begin{eqnarray*}
%u^{B}_{i}(v^{B}_{i};\cM) \triangleq  \mathop\bE_{v^{B}_{-i}\sim D^{B}_{-i};v^{S}\sim D^{S}} u^{B}_{i}(v^{B}_{i};\cM,v^{B}_{-i},v^{S}) \geq\\
% \mathop\bE_{v^{B}_{-i}\sim D^{B}_{-i};v^{S}\sim D^{S}}
%\Big[ \sum_{A\subseteq M} x^{B}_{iA}(v'^{B}_{i},v^{B}_{-i},v^{S}) v^B_i(A) - p^{B}_{i}(v'^{B}_{i},v^{B}_{-i},v^{S}) \Big];
%\end{eqnarray*}
\begin{eqnarray*}
u^{B}_{i}(v^{B}_{i};\cM) &\triangleq&  \mathop\bE_{v^{B}_{-i}\sim D^{B}_{-i};v^{S}\sim D^{S}} u^{B}_{i}(v^{B}_{i};\cM,v^{B}_{-i},v^{S})\\
&\geq& \mathop\bE_{v^{B}_{-i}\sim D^{B}_{-i};v^{S}\sim D^{S}}
\Big[ \sum_{A\subseteq M} x^{B}_{iA}(v'^{B}_{i},v^{B}_{-i},v^{S}) v^B_i(A) 
- p^{B}_{i}(v'^{B}_{i},v^{B}_{-i},v^{S}) \Big];
\end{eqnarray*}
and (2)  for any seller $j$ and values $v^{S}_{j}$, $v'^{S}_{j}$,
\begin{eqnarray*}
 u^{S}_{j}(v^{S}_{j};\cM)
&\triangleq& \mathop\bE_{v^{B}\sim D^{B};v^{S}_{-j}\sim D^{S}_{-j}} u^{S}_{j}(v^{S}_{j};\cM,v^{B},v^{S}_{-j}) \\
&\geq& \mathop\bE_{v^{B}\sim D^{B};v^{S}_{-j}\sim D^{S}_{-j}}
\Big[ p^{S}_{j}(v^{B}, v'^{S}_{j},v^{S}_{-j}) - v^{S}_{j}x^{S}_{j}(v^{B}, v'^{S}_{j},v^{S}_{-j}) \Big].
\end{eqnarray*}

% \\
%&& \hspace{75pt}

Mechanism $\cM$ is {\em Bayesian individually rational} (BIR) if (1) for any buyer $i$ and valuation function $v^{B}_{i}$,
$u^{B}_{i}(v^{B}_{i};\cM) \geq 0$;
and (2)  for any seller $j$ and value $v^{S}_{j}$,
$ u^{S}_{j}(v^{S}_{j};\cM) \geq 0$.

Finally, we denote by $OPT(\cI)$ the (expected) profit generated by the optimal DSIC mechanism for instance $\cI$.

%\paragraph{Production-Cost Markets.}
A special case of two-sided markets is {\em production-cost markets},
where the broker can produce the items by himself
 and each item $j\in  M$ has a publicly known production cost $c_{j}\in \bR^{+}\cup\{0\}$.
Therefore we do not need to consider the sellers' incentives.
%or we can regard sellers' true values are publicly known.
%More precisely, the cost for producing item $j$ is
Letting $c\triangleq (c_{j})_{j\in  M}$,
%and we assume the production costs are  additive.
we use $\cI^{c} = ( N, M,D^{B},c)$ to denote a production-cost market instance
and $\cM^{c}=(x^{B},p^{B})$ a production-cost market mechanism,
where the input of $x^{B}$ and $p^{B}$ is the buyers' valuation profile.
Then the broker's profit is the revenue minus the total production cost
$PFT(\cM^c;\cI^c)$, which is
$$
\mathop{\bE}\limits_{v^{B}\sim D^{B}}
\sum_{i\in  N} \left(p^{B}_{i}(v^{B})  - \sum_{A\subseteq M} \sum_{j\in  A} x^{B}_{iA}(v^{B})c_j\right) .
$$

%Here each item $j$ has a production cost $c_{ij}$ if it is sold to buyer $i$ (we can also regard $c_{ij}$ as the transportation cost).

%\paragraph{Auctions.}
%An even more special case of two-sided markets is auctions,
%where the broker holds all the items and sells the items to buyers without any production cost.
%That is, we can
Auctions are production-cost markets with cost 0.
%so the broker's profit is just the revenue.
% generated from buyers.
We use $\cI^{a} = ( N, M, D^{B})$ to denote an auction instance
and $\cM^{a}=(x^{B},p^{B})$ an mechanism.
%The solution concepts and utility functions of production-cost markets
%and auctions are the same with the general two-sided markets.
%When the sellers are not included in the markets,
%a mechanism for production-cost markets or auctions only includes
%an allocation function and a payment function of the buyers $\cM=(x^{B},p^{B})$
%where the input of these two functions is the buyers' valuation profile.
The expected revenue is
%the profit is indeed its revenue which is denoted by
$PFT(\cM^a;\cI^a) = \bE_{v^{B}\sim D^{B}} \sum_{i\in  N} p^{B}_{i}(v^{B})$.
When there is no ambiguity, the superscript $B$ is omitted in auctions and production-cost markets.

%Finally, denote by $OPT(\cI)$ the revenue generated by the optimal mechanism.

In Section \ref{sec:cost:additive},
we will consider additive valuations for the buyers.
In this case,
%Given a valuation function $v^{B}_{i}$ for buyer~$i$, if $v^{B}_{i}$ is additive or unit-demand,
for any buyer $i$, there exists a valuation vector
$(v^{B}_{ij})_{j\in  M}$ such that $v^{B}_{ij} = v^{B}(\{j\})$ is $i$'s value on each item~$j$.
Then, $v^{B}_i$ is {\em additive} if $v^{B}_{i}(A)=\sum_{j\in A} v^{B}_{ij}$ for any $A\subseteq  M$.
%and $v^{B}_i$ is {\em unit-demand} if $v^{B}_{i}(A)=\max_{j\in A} v^{B}_{ij}$ for any $A\subseteq  M$.
To simplify the notation, in this case
we use $v^{B}_{i}$ to denote the vector $(v^{B}_{ij})_{j\in M}$ instead of the corresponding function.
Each $v^{B}_{ij}$ is independently drawn from a distribution $D^{B}_{ij}$,
and $D^{B}_{i} = \times_{j\in M} D^{B}_{ij}$.
Finally, when buyers have additive valuations,
their allocation is simplified as
$x^{B}(v^{B},v^{S})\triangleq(x^{B}_{i}(v^{B},v^{S}))_{i\in  N}$,
where $x^{B}_{i}(v^{B},v^{S})=(x^{B}_{ij}(v^{B},v^{S}))_{j\in  M}$
with $x^{B}_{ij}(v^{B},v^{S})\in [0,1]$, representing the probability that buyer~$i$ gets the item $j$,
when the valuations are $v^{B}$ and $v^{S}$.
%Moreover, $\sum_A x^{B}_{iA} (v^{B},v^{S})= 1$.

\section{A Reduction from Two-sided markets to Production-Cost Markets}
\label{sec:monotone}
%As we have seen in our first reduction, the production-cost market is closely related to the 1-sided market.
Note that the sellers are single-parameter in the two-sided markets under consideration.
Thus, each seller is truthful in a mechanism if and only if the selling probability of her item
is non-increasing with respect to her value and
the payment to her is the threshold payment, i.e., the highest value such that her item can still be sold.
More precisely,
for any single-value distribution $D$ with density function $f$ and cumulative probability $F$,
%if $D$ is a buyer's valuation distribution, then the virtual value function is $\phi^{B}(v)=v-\frac{1-F(v)}{f(v)}$
if $D$ is a seller's value distribution, then the virtual value function is $\phi^{S}(v)=v+\frac{F(v)}{f(v)}$.
In addition, if $D$ is not regular then $\phi^S$ is the ironed virtual value.
Following \cite{myerson1983efficient}, for single-parameter sellers and
any DSIC mechanism $\cM=(x^{S},x^{B}, p^{S}, p^{B})$,
the total payment to the sellers is
the virtual social welfare of them, i.e.,
\begin{equation}\label{equ:myerson}
\mathop\bE\limits_{v^S\sim D^S} \sum_{j\in  M} p^S_j(v^B, v^S) = \mathop\bE\limits_{v^S\sim D^S} \sum_{j\in M}  \phi_j(v^S_j) x^{S}_{j}(v^{B}, v^{S})
\end{equation}
for any valuation profile $v^B$ of the buyers.

%With this property, w
We now show how to convert a mechanism for production-cost markets into a two-sided market's mechanism.
The main idea is to use the sellers' virtual values in two-sided markets as costs,
and run the mechanism for production-cost markets.

\begin{definition}
%{\bf (Cost monotonicity)}
A mechanism $\cM^{c} = (x, p)$ for production-cost markets is {\em cost-monotone} if
for any two instances $\cI^{c}=(N,M,D^{c}, c)$ and $\cI'^{c}=(N,M,D^{c}, c')$,
where $c$ and $c'$ differ only at an item~$j$ and $c_j\leq c_j'$,
for any buyers' valuation profile $v^{c}\sim D^{c}$,
the probabilities of item $j$ being sold under the two instances,
$x_j \triangleq \sum_{i\in N} \sum_{A\ni j} x_{iA}(v^c; \cI^c)$
and $x'_j \triangleq \sum_{i\in N} \sum_{A\ni j} x_{iA}(v^c; \cI'^c)$,
satisfy $x_j\geq x'_j$.
%is non-increasing in~$c_{j}$.
\end{definition}

{\bf Reduction.}
Let $\cI= (N, M, D^S, D^B)$ be a two-sided market instance.
For any valuation profile $v^{S}$ of the sellers,
%\sim D^S$ for the sellers,
denote by $\phi^{S}(v^{S}) \triangleq (\phi^{S}_j(v^{S}_{j}))_{j\in M}$  the sellers' virtual-value vector,
and let $\cI^{c}_{\phi^S(v^S)} = (N, M, D^B, \phi^{S}(v^S))$ be a production-cost market instance.
% where $\phi^{S}(v^S)$ is the cost vector.

%Before we design two-sided market mechanism using the production-cost market mechanism,
We first show that the optimal profit of the two-sided market is no more than
the optimal profit generated by the corresponding production-cost markets in expectation.

\begin{lemma}
\label{lem:reduction:2:1}
For any two-sided market instance $\cI=(N,M,D^{B},D^{S})$,
$OPT(\cI) \leq \bE_{v^S \sim D^S} OPT(\cI_{\phi^S(v^S)}^{c})$.
\end{lemma}
\begin{proof}

It suffices to show that for any DSIC mechanism $\cM=(x^{S},x^{B}, p^{S}, p^{B})$ for two-sided markets,
there exists a DSIC mechanism $\cM^{c}$ for production-cost markets such that
$PFT(\cM;\cI) \leq \bE_{v^{S}\sim D^{S}} PFT(\cM^{c};\cI^{c}_{\phi^S(v^{S})})$.
Indeed, this would imply $PFT(\cM;\cI) \leq \bE_{v^{S}\sim D^{S}} OPT(\cI^{c}_{\phi^S(v^{S})})$ for any $\cM$,
and thus $OPT(\cI)\leq \bE_{v^{S}\sim D^{S}} OPT(\cI^{c}_{\phi^S(v^{S})})$.

Given $\cM$ and $\cI$, we define mechanism $\cM^{c}=(x^{c}, p^{c})$ as follows.
For any instance $\cI^{c}_{\phi^S(v^{S})}$,
% = (N, M, D^B, \phi^{S}(v^S))$,
$\cM^c$ first computes~$v^S$, the (randomized) {\em pre-image} of $\phi^S(v^S)$ with respect to $D^S$.
In particular,
if for some seller $j$, the (ironed) virtual value $\phi^S_j(v^S_j)$
corresponds to a value interval in the support of $D^S_j$,
then $v^S_j$ is randomly sampled from $D^S_j$ conditional on it belongs to this interval.

%In particular, if $D^S$ is regular then $v^S$ is uniquely defined;

For any reported valuation profile $v^{B}$ and buyer $i\in N$,
%randomly draw $\bar{v}^{S}\sim D^{S}$ and~let
$$
x^{c}_{iA}(v^{B}) = x^{B}_{iA}(v^{B},v^{S})
$$
for any $A\subseteq M$, and
$$
p^{c}_{i}(v^{B}) = p^{B}_{i}(v^{B}, v^{S}).
$$
It is easy to see that, given any $v^S$ and $v^B_{-i}$, for any true valuation~$v^B_i$,
buyer $i$ has the same utility in $\cM^{c}$ and $\cM$ by reporting the same~$v'^B_i$.
Thus $\cM^{c}$ is DSIC whenever $\cM$ is DSIC.
%Here we note that if $\cM$ is BIC, $\cM^{c}$ may not be BIC
%since $v^{S}$ is a fixed valuation profile for sellers.
Next, we lower-bound the profit of $\cM^{c}$ for each instance $\cI^{c}_{\phi^S(v^{S})}$.

%\begin{eqnarray*}
%PFT(\cM^{c}_{v^{S}},\cI_{v^{S}}) &=& \mathop\bE\limits_{v^B \sim D^B}
%\sum_{i\in  N} (p^B_i(v^B, v^S)\\
%&&- \sum_{A \subseteq M} \sum_{j \in A} x^B_{iA}(v^B, v^S) \phi_j(v^S_j))
%\end{eqnarray*}

\begin{eqnarray*}
&&PFT(\cM^{c}; \cI^{c}_{\phi^S(v^{S})}) \\
&=&
\mathop\bE\limits_{v^B \sim D^B}
\sum_{i\in  N} \left(p^{c}_i(v^B)
- \sum_{A \subseteq M} x^{c}_{iA}(v^B)  \sum_{j \in A}\phi^S_j(v^S_j)\right)\\
&=&
\mathop\bE\limits_{v^B \sim D^B} \mathop\bE\limits_{v^S\sim D^S|\phi^S(v^S)}
\left(\sum_{i\in  N} p^B_i(v^B, v^{S})
- \sum_{j\in M}  \sum_{i \in N} \sum_{A \ni j} x^B_{iA}(v^B, v^{S})\phi^S_j(v^S_j) \right) \\
&\geq &
\mathop\bE\limits_{v^B \sim D^B} \mathop\bE\limits_{v^S\sim D^S|\phi^S(v^S)}
\left(\sum_{i\in  N} p^B_i(v^B, v^{S}) 
- \sum_{j\in M}\phi^S_j(v^S_j)
  x^{S}_{j}(v^{B},v^{S})\right)
\end{eqnarray*}

The inequality above is because 
$\sum_{i\in  N}\sum_{A\ni j} x^{B}_{iA}(v^{B},v^{S}) \leq x^{S}_{j}(v^{B},v^{S})$ holds
for any feasible mechanism, any $j\in  M$ and any valuation profiles $v^{B},v^{S}$.
Thus,
%{\small
%\begin{eqnarray*}
%&& \mathop\bE\limits_{v^S \sim D^S}  PFT(\cM^{c}; \cI^{c}_{\phi^S(v^{S})}) \\
%&=& \mathop\bE\limits_{\phi^S(v^S)\sim \phi^S(D^S)}  PFT(\cM^{c}; \cI^{c}_{\phi^S(v^{S})}) \\
%&\geq &
%\mathop\bE\limits_{v^B \sim D^B} \mathop\bE\limits_{v^{S} \sim D^S}
%\left(\sum_{i\in  N} p^B_i(v^B, v^{S}) - \sum_{j\in M}  \phi_j(v^S_j) x^{S}_{j}(v^{B}, v^{S})\right)\\
%&=&
%\mathop\bE\limits_{v^S \sim D^S, v^B \sim D^B}
%\left(\sum_{i\in  N} p^B_i(v^B, v^S) - \sum_{j\in  M} p^S_j(v^B, v^S)\right)\\
%&=&PFT(\cM,\cI),
%\end{eqnarray*}
%}

\begin{eqnarray*}
&& \mathop\bE\limits_{v^S \sim D^S}  PFT(\cM^{c}; \cI^{c}_{\phi^S(v^{S})}) 
= \mathop\bE\limits_{\phi^S(v^S)\sim \phi^S(D^S)}  PFT(\cM^{c}; \cI^{c}_{\phi^S(v^{S})}) \\
&\geq &
\mathop\bE\limits_{v^B \sim D^B} \mathop\bE\limits_{v^{S} \sim D^S}
\left(\sum_{i\in  N} p^B_i(v^B, v^{S}) 
- \sum_{j\in M}  \phi_j(v^S_j) x^{S}_{j}(v^{B}, v^{S})\right)\\
&=&
\mathop\bE\limits_{v^S \sim D^S, v^B \sim D^B}
\left(\sum_{i\in  N} p^B_i(v^B, v^S) - \sum_{j\in  M} p^S_j(v^B, v^S)\right)
=PFT(\cM,\cI),
\end{eqnarray*}
as desired. Here $\phi^S(D^S)$ is the distribution of virtual values induced by $D^S$, and the second equality is by Equation \ref{equ:myerson}.
\end{proof}

In the following, we show that if a mechanism for production-cost markets is cost-monotone,
then it can be converted into a mechanism for two-sided markets. % and the profit does not lose.

\begin{lemma}
\label{lem:reduction:2:2}
Given any DSIC cost-monotone mechanism $\cM^{c}$ for production-cost markets,
there exists a DSIC mechanism $\cM$ for two-sided markets such that
$$PFT(\cM; \cI) = \bE_{v^{S}\sim D^{S}} PFT(\cM^{c}; \cI^{c}_{\phi^S(v^{S})}).$$
\end{lemma}

\begin{proof}
Given mechanism $\cM^{c}=(x^{c}, p^{c})$,
the mechanism $\cM=(x^{S},x^{B}, p^{S}, p^{B})$ is defined as follows:
$\cM$ first collects $v^{B}$ and $v^{S}$ reported by the buyers and the sellers, and then
 run $\cM^{c}$
on the production-cost instance $\cI^{c}_{\phi^S(v^S)} = (N, M, D^B, \phi^{S}(v^S))$ to obtain $x^{c}(v^B)$ and $p^{c}(v^B)$.
Then for each buyer $i$, let
$$x^{B}_{iA}(v^B, v^S)=x^{c}_{iA}(v^B)$$
for any $A\subseteq M$ and
$$p^{B}_{i}(v^B, v^S)=p^{c}_{i}(v^B).$$
For each seller $j$, let
$$x^{S}_{j}(v^{B},v^{S}) = \sum_{i\in  N} \sum_{A \ni j} x^{c}_{iA}(v^S, v^B)$$
and let $p^S_j(v^B, v^S)$ be the threshold payment for $j$:
namely, the highest reported value of seller $j$ such that the probability that
item~$j$ is bought by the broker is $x^{S}_{j}(v^{B},v^{S})$.
%$$p^{S}_{j}(v^B, v^S) = x^{S}_{j}(v^{B},v^{S})v^{S}_{j} + \int_{v^{S}_{j}}^{+\infty} x^{S}_{j}((v^{S}_{-j}, v_{j}),v^{B}) dv_{j}.$$

% runs $\cM'$ on $\cI_{t^S}$ to get
%$x^{B*}_{iS}(t^S, t^B)$ and $p^{B*}_{i}(t^S, t^B)$.
%Let $x^{S*}_{j}(t^S, t^B) = \sum_{i\in N} \sum_{S \ni j} x^{B*}_{iS}(t^S, t^B)$ and let the seller pays threshold payment.
We claim that $\cM$ is DSIC.
First, the buyers will truthfully report their valuations because $\cM^{c}$ is DSIC
and each buyer has the same allocation and payment in $\cM$ and $\cM^{c}$.
For the sellers,
since $\cM^{c}$ is cost-monotone and each (ironed) virtual value function $\phi^S_j$ is non-decreasing in $v^S_j$,
the allocation $x^S_j$ is non-increasing in~$v^S_j$.
As the payments to the sellers are the threshold payments,
the sellers are truthful as well.
% is each seller's dominant strategy.

Next we show that %$PFT(\cM; \cI) = \bE_{v^{S}\sim D^{S}} PFT(\cM^{c}; \cI^{c}_{\phi^S(v^{S})})$.
\begin{eqnarray*}
&&PFT(\cM, \cI) 
= \mathop\bE\limits_{v^S \sim D^S, v^B \sim D^B}
\left(\sum_{i\in  N} p^{B}_i(v^B, v^S) - \sum_{j\in  M} p^{S}_j(v^B, v^S) \right) \\
&=&\mathop\bE\limits_{v^S \sim D^S, v^B \sim D^B}
\left(\sum_{i\in  N} p^{B}_i(v^B, v^S) 
- \sum_{j \in  M} x^{S}_{j}(v^{B},v^{S}) \phi_j(v^S_j)\right) \\
&=&\mathop\bE\limits_{v^B \sim D^B} \mathop\bE\limits_{\phi^S(v^S)\sim \phi^S(D^S)}
\left(\sum_{i\in  N} p^{c}_i(v^B) 
- \sum_{j \in  M}  \sum_{i\in  N} \sum_{A \ni j} x^{c}_{iA}(v^B) \phi_j(v^S_j)\right) \\
&=& \mathop\bE\limits_{\phi^S(v^S)\sim \phi^S(D^S)}
\mathop\bE\limits_{v^B \sim D^B}
\sum_{i\in  N} \left(p^{c}_i(v^B) 
- \sum_{A\subseteq M} \sum_{j \in A} x^{c}_{iA}(v^B) \phi_j(v^S_j)\right) \\
&=& \mathop\bE\limits_{v^S \sim D^S} PFT(\cM^{c}; \cI^{c}_{\phi^S(v^{S})}).
\end{eqnarray*}
Thus Lemma \ref{lem:reduction:2:2} holds.
\end{proof}

Combining Lemmas \ref{lem:reduction:2:1} and \ref{lem:reduction:2:2},
we get our first main result.
\begin{theorem}\label{thm:reduction:2}
Given any DSIC mechanism $\cM^{c}$ for production-cost markets,
if $\cM^{c}$ is cost-monotone and is an $\alpha$-approximation to the optimal profit,
then there exists a DSIC mechanism
$\cM$ for two-sided markets that is an $\alpha$-approximation to the optimal profit.
\end{theorem}

\begin{proof}
Mechanism $\cM$ is defined as in Lemma \ref{lem:reduction:2:2}.
For any two-sided market instance $\cI$,
%Since $\cM^{c}$ is $\alpha$-approximation, $PFT(\cM^{c}, \cI^{c}_{v^{S}}) \geq \frac{1}{\alpha}OPT(\cI^{c}_{v^{S}})$.
\begin{eqnarray*}
PFT(\cM; \cI) &=& \mathop\bE\limits_{v^{S}\sim D^{S}} PFT(\cM^{c}; \cI^{c}_{\phi^S(v^{S})}) 
\geq \frac{1}{\alpha}\mathop\bE\limits_{v^S \sim D^S} OPT(\cI_{\phi^S(v^S)}^{c}) \geq \frac{1}{\alpha}OPT(\cI),
\end{eqnarray*}
where the equality is by Lemma \ref{lem:reduction:2:2} and
the last inequality is by Lemma~\ref{lem:reduction:2:1}.
\end{proof}

%\section{Linear Programming Representation}
%
%
%We then write the linear programming of the optimization of the revenue.
%For any mechanism $\cM=(x,x^{S},p,p^{S})$, any buyer $i$ with  and seller $j$, let
%(1) $x_{ij}(v_{i})$ be the expected probability that buyer $i$ with value $v_{i}$ gets the item $j$
%(over the randomness of the mechanism and $v_{-i}$ and $v^{S}$);
%(2) $x^{S}_{j}(v^{S}_{j})$ be the expected probability that seller $j$ with value $v^{S}_{j}$ sells her item
%(over the randomness of the mechanism and $v$ and $v^{S}_{-j}$);
%(3) $p_{i}(v_{i})$ be the expected payment to buyer $i$ with value $v_{i}$
%(over the randomness of the mechanism and $v_{-i}$ and $v^{S}$);
%(4) $p^{S}_{j}(v^{S}_{j})$ be the expected payment from seller $j$ with value $v^{S}_{j}$
%(over the randomness of the mechanism and $v$ and $v^{S}_{-j}$).
%
%Then the goal is to maximize the revenue.
%Let $T^{B+}_{ij}=T\cup\{\emptyset\}$ and $T^{S+}_{j}=T^{S}\cup\{\emptyset\}$.
%Moreover, $x_{ij}(\emptyset)=0$, $x^{S}_{j}(\emptyset)=0$, $p_{i}(\emptyset)=0$ and $x^{S}_{j}(\emptyset)=0$.

\section{A Mechanism for Two-Sided Markets with Additive Valuations}
\label{sec:cost:additive}

\subsection{Broker's Profit in Production-Cost Markets}
We first design a mechanism $\cM_{A}$ for production-cost markets
which is an 8-approximation of the optimal profit.
Our mechanism is inspired by the mechanism in \cite{yao2015n} and
the duality framework in \cite{cai2016duality} for auctions.
In particular, with probability $\frac{3}{4}$, $\cM_{A}$ runs the mechanism of \cite{myerson1983efficient} for two-sided markets for each item separately, denoted by $\cM_{IT}$. % (for ``individual two-sided'').
The mechanism of \cite{myerson1983efficient} is for a single buyer and a single seller,
but can be generalized to multiple buyers and a single seller as shown in \cite{deng2014revenue}.
%And a mechanism for two-sided markets can certainly be run for production-cost markets, because the latter are special cases of the former.
%To be sure, when we convert our mechanism for production-cost markets
Furthermore, $\cM_A$ generalizes
the bundling VCG mechanism of \cite{yao2015n}
to production-cost markets (denoted by $\cM_{BVCG}$) and runs it with probability $\frac{1}{4}$.
%
%In the following, we first describe these two sub-mechanisms and then
%use them to upper-bound the optimal profit.

% is easy to describe.
%Since production-cost markets are special cases of two-sided markets,
%the mechanism of \cite{myerson1983efficient} directly applies here for each item.
Essentially, Mechanism $\cM_{IT}$ runs a second-price auction on the buyers' virtual values, with a reserve price which is the production cost of the item.
As shown in \cite{myerson1983efficient,deng2014revenue}, this mechanism is optimal for the broker's profit when the buyers have single-parameter valuations.
Mechanism $\cM_{BVCG}$ is well studied in auctions \cite{yao2015n,cai2016duality},
and we describe it in Mechanism \ref{alg:bvcg} for production-cost markets $\cI^c = (N, M, D, c)$.
Essentially, it is a VCG mechanism with per-item reserve prices and per-agent entry fees.

\begin{algorithm}[htbp]
\floatname{algorithm}{Mechanism}
  \caption{\hspace{-3pt}  $\cM_{BVCG}$ for Production-Cost Markets}
 \label{alg:bvcg}
  \begin{algorithmic}[1]
%\REQUIRE Distribution $D$ and value $v$.
\STATE Collect the valuation profile $v$ from the buyers.
\STATE For any buyer $i$ and item $j$,
% and $v_{-i}$,
 let $P_{ij}(v_{-i}) \triangleq \max_{i'\neq i} v_{i'j}$
and $\beta_{ij}(v_{-i}) \triangleq \max\{P_{ij}(v_{-i}), c_{j}\}$. \label{step:beta}

\STATE For any buyer $i$,
set the reserve price for item $j$ to be $\beta_{ij}(v_{-i})$.
Set the entry fee $e_i(v_{-i})$ to be the median of the random variable $\sum_{j\in  M}(t_{ij} - \beta_{ij}(v_{-i}))^{+}$,
where $t_{i}=(t_{ij})_{j\in M}\sim D_{i}$ and $x^+\triangleq \max\{x, 0\}$ for any $x\in \bR$. \label{step:entry-fee}

\STATE Each buyer $i$ is considered to accept her entry fee if and only if
 $\sum_{j\in  M}(v_{ij} - \beta_{ij}(v_{-i}))^{+}\geq e_i(v_{-i})$.

\STATE
If a buyer $i$ accepts her entry fee, then she gets the set of items $j$ with $v_{ij} \geq \beta_{ij}(v_{-i})$,
and her price is
$e_i(v_{-i}) + \sum_{j: v_{ij} \geq \beta_{ij}(v_{-i})} \beta_{ij}(v_{-i})$.
If $i$ does not accept her entry fee, then she gets no item and pays 0.
%\STATE For any seller $j$, for player $i^*$ such that $v_{i^*j} \geq P_{i^*j}(v_{-i^*})$, calculate value
%$a = \argmax\limits_{a} \{a | \argmax\limits_{p \geq a}
%\sum_j (v_{i^*j} - \beta_{i^*j}(P_{ij}(v_{-i^*}), v^{S}))^+
%\geq e_i(p_{i,-j}, \max\{P_{i^*j}(v_{-i^*}), a\})\}$.
%Post price $p^S_j = \phi^{-1}(a)$.

%\STATE Sell the item $j$ to buyer $i$ if and only if
%she accepts the entry fee $e_i(v_{-i})$ and $v_{ij} \geq \beta_{ij}(v_{-i})$. % and $v_{ij} \geq p_{ij}$ and $c_j \leq p^S_j$.

%\ENSURE
\end{algorithmic}
\end{algorithm}

It is not hard to see that both $\cM_{IT}$ and $\cM_{BVCG}$
are DSIC and IR. Indeed, the mechanism of \cite{myerson1983efficient} is DSIC and IR,
$\cM_{IT}$ directly applies it to each item, and the buyers have additive valuations
across the items.
Moreover, $\cM_{BVCG}$ is DSIC and IR with respect to any reserve prices $\beta_{ij}$
that do not depend on $v_{ij}$, and Mechanism  \ref{alg:bvcg} simply incorporates
the production costs into reserve prices.

%is the mechanism of \cite{myerson1983efficient}
%directly applied to each item,

In Theorem \ref{thm:cost} we use $\cM_{A}$ to upper-bound the optimal profit for any
production-cost instance $\cI^c = (N, M, D, c)$, with proof provided in the full version.
%Due to limit of space, we prove it in the full version.
In fact, this proof is similar to the proof in \cite{cai2016duality}
with modifications to incorporate the production costs into consideration.
Note that \cite{cai2017approximating} also adapts the framework of \cite{cai2016duality} to the 2-sided market.
But their goal is to maximize the gain from trade and the buyers have unit-demand valuations.

%In fact, we use $\cM_{A}$ to upper-bound the optimal profit of BIC-BIR mechanisms for $\cI^c$,
%which is an upper-bound for $OPT(\cI^{c})$.
\begin{theorem}[\cite{cai2017approximating}]
\label{thm:cost}
When the buyers have additive valuations,
Mechanism $\cM_{A}$ is DSIC and is an 8-approximation to the optimal profit for production-cost markets.
\end{theorem}

\subsection{Converting  $\cM_{A}$ to Two-sided Markets}
Next we prove the cost-monotonicity for Mechanism $\cM_{A}$.
First, we start with Mechanism $\cM_{IT}$.

\begin{lemma}
\label{lem:singlemono}
$\cM_{IT}$ is cost-monotone.
\end{lemma}
\begin{proof}
For any two production-cost instances $\cI^{c}=( N, M,D,c)$ %and an item $j\in  M$
and $\cI'^{c}=( N, M,D,c')$,
where there exists an item $j\in  M$ such that $c'_{j}>c_{j}$ and $c'_{j'} = c_{j'}$ for any $j' \neq j$,
we show that in Mechanism $\cM_{IT}$, when buyers' valuation profile is $v\sim D$,
if item $j$ is not sold in $\cI^{c}$,
then item $j$ is not sold in $\cI'^{c}$.
Since all buyers' valuation functions are additive and $\cM_{IT}$ sells each item individually,
the result of selling one item does not effect any other item.
In the mechanism of \cite{myerson1983efficient},
given the reported valuation profile $v$,
the potential winner of item $j$ is the buyer who has highest virtual value on it, denoted by $i_{j} = \arg\max_{i\in  N} v_{ij}$.
If her virtual value $\phi_{i_j j}(v_{i_j j})$ is at least the cost of item $j$, buyer $i_{j}$ takes item $j$.
Otherwise, item $j$ is kept unsold.
Therefore, if item $j$ is not sold in $\cI^{c}$, then $\phi_{i_j j}(v_{i_j j}) - c_j  < 0$
which implies $ \phi_{i_j j}(v_{i_j j}) - c'_j<0$ and item $j$ cannot be sold in $\cI'^{c}$.
Thus $\cM_{IT}$ satisfies cost-monotonicity.
\end{proof}

Next we show $\cM_{BVCG}$ is cost-monotone.
Since we need to apply $\cM_{BVCG}$ to different instances with different cost vectors $c$ and $c'$,
we explicitly write $\beta_{ij}(v_{-i}, c_j)$ and $e_{i}(v_{-i}, c)$
in Steps \ref{step:beta} and~\ref{step:entry-fee} of Mechanism~\ref{alg:bvcg}.
%It is easy to see that $BVCG^{c}$ is a VCG mechanism with non-anonymous item reserves and a entry fee for each buyer.
%%works as follows:
%Given a production-cost instance $\cI^{c}=( N, M,D,c)$,
%$BVCG^{c}$ solicits bids $v = (v_{1},\cdots,v_{n})$ from the buyers.
%Then for each buyer $i$, $BVCG^{c}$ offers her the option to participate the VCG mechanism with item reserve $c_{j} + \max_{k\neq i}(v_{kj}-c_{j})^{+}$ for each item $j$
%and an entry fee $e_{i}(v_{-i},D_{i},c)$,
%which is the median of the random variable $\sum_{j\in  M}(v_{ij}-c_{j} - \max_{k\neq i}(v_{kj}-c_{j})^{+})^{+}$,
%where $v_{i}\sim D_{i}$.

\begin{lemma}
\label{lem:additive}
$\cM_{BVCG}$ is cost-monotone.
\end{lemma}

\begin{proof}

Similarly, for any two production-cost instances $\cI^{c}=( N, M,D,c)$ %and an item $j\in  M$
and $\cI'^{c}=( N, M,D,c')$,
where there exists an item $j\in  M$ such that $c'_{j}>c_{j}$ and $c'_{j'} = c_{j'}$ for any $j' \neq j$,
we show that in Mechanism $\cM_{BVCG}$, when buyers' valuation profile is $v\sim D$,
if item $j$ is sold in $\cI'^{c}$,
then item $j$ is also sold in $\cI^{c}$.

In Mechanism $\cM_{BVCG}$, given
the valuation profile $v$, the potential winner of item $j$ is the buyer who has highest value on item~$j$,
denoted by $i_{j} = \arg\max_{i\in  N} v_{ij}$.
%letting $i_{j}(v,c)$ be the buyer who has highest reported value on item $j$ in $\cM_{BVCG}$.
%$ \arg \max_{i\in  N} (v_{ij} - c_{j})^{+}$,
%$BVCG^{c}$ always keeps a winning player $i\in N$ who has the maximum value on item $j$ with cost $c_{j}$
When the cost vector is $c$, item~$j$ is sold to $i_j$
if and only if $i_{j}$ accepts the entry fee $e_{i_{j}}(v_{-i},c)$ and $v_{i_j} - \beta_{i_j j}(v_{-i_j}, c_j)>0$.
Otherwise item $j$ is unsold.
Note that given different production cost $c'$, the potential winner of item $j$ remains unchanged.
%If $\max_{i\in  N} (v_{ij} - c_{j})^{+}>0$, $i_{j}(v,c)=\arg \max_{i\in  N} (v_{ij} - c_{j})^{+}$;
%If $\max_{i\in  N} (v_{ij} - c_{j})^{+}=0$, item $j$ is kept unsold and set $i_{j}(v,c)=0$.
%Thus $i_{j}(v,c')=i$ for some $i\in N$ implies $i_{j}(v,c)=i$
%and $i_{j}(v,c)=0$ implies $i_{j}(v,c')=0$.

Note that the entry fee $e_{i}(v_{-i},c)$ is selected such that the probability that buyer $i$ accepts it is exactly $\frac{1}{2}$.
% $\beta_{i_j j}(v_{-i_j}, c_j)\triangleq\max\{P_{i_{j}j}(v_{-i_j}), c_j\}$,
%$\beta_{i_j j}(v_{-i_j}, c'_j)\triangleq\max\{P_{i_{j}j}(v_{-i_j}), c'_j\}$
Let $d_j \triangleq \beta_{i_j j}(v_{-i_j}, c'_j) - \beta_{i_j j}(v_{-i_j}, c_j)$
be the increase of the item reserve in $\cM_{BVCG}$ for buyer $i_j$.
Then
\begin{equation}
\label{eq:bvcg:monotone}
(v_{i_{j}j}-\beta_{i_{j}j}(v_{-i_{j}}, c'_j))^{+} + d_j \geq (v_{i_{j}j}-\beta_{i_{j}j}(v_{-i_{j}}, c_j))^{+}.
\end{equation}
Indeed the equality holds in Inequality \ref{eq:bvcg:monotone} if $v_{i_{j}j}-\beta_{i_{j}j}(v_{-i_{j}}, c'_j)\geq 0$.

Let $e'_{i_{j}} \triangleq e_{i_{j}}(v_{-i_{j}},c) - d_j$. % and for any buyers' value $v$,
%for buyer $i$, if she accepts the entry fee $e_{i}$ under cost $c$,
When the cost vector is $c$, buyer $i_{j}$'s utility is
$$u_{i_{j}} = \sum_{k\in M} (v_{i_{j}k}-\beta_{i_{j}k}(v_{-i_{j}}, c_k))^{+}-e_{i_{j}}(v_{-i_{j}},c).$$
When the cost is $c'$, if the entry fee is $e'_{i_{j}}$, buyer $i_{j}$'s utility is
\begin{eqnarray*}
&&u'_{i_{j}} (e'_{i_{j}}) = \sum_{k\in M} (v_{i_{j}k}-\beta_{i_{j}k}(v_{-i_{j}}, c'_k))^{+}-e'_{i_{j}}\\
&=&\sum_{k \neq j} (v_{i_{j}k}-\beta_{i_{j}k}(v_{-i_{j}}, c_k))^{+} + (v_{i_{j}j}-\beta_{i_{j}j}(v_{-i_{j}}, c'_j))^{+} -e'_{i_{j}}\\
&=& \sum_{k \neq j} (v_{i_{j}k}-\beta_{i_{j}k}(v_{-i_{j}}, c_k))^{+} + (v_{i_{j}j}-\beta_{i_{j}j}(v_{-i_{j}}, c'_j))^{+}\
-e_{i_{j}}(v_{-i_{j}},c) + d_j \\
&\geq& \sum_{k \neq j} (v_{i_{j}k}-\beta_{i_{j}k}(v_{-i_{j}}, c_k))^{+} + (v_{i_{j}j}-\beta_{i_{j}j}(v_{-i_{j}}, c_j))^{+} 
- e_{i_{j}}(v_{-i_{j}},c)
= u_{i_{j}}.
\end{eqnarray*}
The inequality above is by Inequality \ref{eq:bvcg:monotone}.
Moreover, if $v_{i_{j}j}-\beta_{i_{j}j}(v_{-i_{j}}, c'_j) \geq 0$, then
\begin{equation}
\label{eq:bvcg:monotone:1}
u'_{i_{j}} (e'_{i_{j}})= u_{i_{j}}.
\end{equation}

That is for any valuation profile $v_{i_{j}}$,
buyer $i_{j}$'s utility under the entry fee $e'_{i_{j}}$ and the cost vector $c'$
is at least her utility under the entry fee $e_{i_{j}}(v_{-i_{j}},c)$ and the cost vector $c$.
Therefore,
$$
\Pr_{t_{i_{j}}\sim D_{i_{j}}}[\sum_{j\in  M}(t_{i_{j}j} - \beta_{i_{j}j}(v_{-i_{j}}))^{+}\geq e'_{i_{j}}] \geq \frac{1}{2}.
$$
%That is, % the probability buyer $i_j$ accepts $e_{i_j}$ under cost $c$ is $\frac{1}{2}$,
%the probability buyer $i_j$ accepts $e'_{i_j}$ with item reserve $\beta_{i_j j}(v_{-i_j}, c'_j)$ is at least $\frac{1}{2}$.
Since the real entry fee $e_{i_j}(v_{-i_{j}},c')$ is selected to be the median of
the random variable $\sum_{j\in  M}(t_{i_{j}j} - \beta_{i_{j}j}(v_{-i_{j}}))^{+}$, we have
$e_{i_j}(v_{-i_{j}},c') \geq e'_{i_{j}}$.

Now if under cost vector $c'$, item $j$ is sold, then
$$v_{i_{j}j}-\beta_{i_{j}j}(v_{-i_{j}}, c'_j) \geq 0,$$
and
$$\sum_{k\in M} (v_{i_{j}k}-\beta_{i_{j}k}(v_{-i_{j}}, c'_k))^{+}-e_{i_j}(v_{-i_{j}},c') \geq 0.$$
Thus under cost vector $c$, we have
$$
v_{i_{j}j}-\beta_{i_{j}j}(v_{-i_{j}}, c_j)\geq 0,
$$
and by Equation \ref{eq:bvcg:monotone:1},
\begin{eqnarray*}
u_{i_{j}}&=&u'_{i_{j}} (e'_{i_{j}}) =\sum_{k\in M} (v_{i_{j}k}-\beta_{i_{j}k}(v_{-i_{j}}, c'_k))^{+}-e'_{i_{j}}\\
&\geq&\sum_{k\in M} (v_{i_{j}k}-\beta_{i_{j}k}(v_{-i_{j}}, c'_k))^{+}-e_{i_j}(v_{-i_{j}},c') \geq 0.
\end{eqnarray*}

Therefore, under cost vector $c$, item $j$ is also sold.
That is, $\cM_{BVCG}$ is cost-monotone.
\end{proof}

By randomly selecting from $\cM_{IT}$ and $\cM_{BVCG}$,
Mechanism $\cM_A$ is still cost-monotone.
Therefore, by Theorem~\ref{thm:reduction:2},
$\cM_{A}$ can be converted into a mechanism for two-sided markets, which is again an 8-approximation to the optimal profit.
Formally, we have the following theorem. 

\begin{theorem}
\label{thm:add}
When the buyers have additive valuations,
there exists a DSIC mechanism that is an 8-approximation to the optimal profit for two-sided markets.
\end{theorem}

\section{Conclusion and Open Problems.}
In this paper we provide the first DSIC mechanism that is a constant-approximation to the broker's optimal profit
in multi-parameter settings with more than one buyer.
%when the buyers' have additive valuations.
We use production-cost markets as a bridge between auctions and two-sided markets, and provide a general reduction from production-cost markets to two-sided markets.
%the former to the latter.
%(1) In the second half of this paper, we present constant approximation mechanisms for the situations where
%buyers are single-parameter, additive or unit-demand.
%When the buyers have sub-additive valuations,
How to design DSIC mechanisms for the broker's profit in multi-buyer multi-parameter settings
%multi-seller multi-buyer markets
under other valuation functions of the buyers (e.g.
unit-demand or sub-additive) is still open,
and it would be interesting to understand the role of production-cost markets in those scenarios.

\bibliographystyle{apalike}
\bibliography{2-side}

\appendix
\section{Proof of Theorem \ref{thm:cost}}

The proof of Theorem \ref{thm:cost} is almost the same with \cite{cai2016duality} 
with modifications to deal with the production cost. 
We provide the full proof here for completeness. 
Similar to \cite{cai2016duality}, we only
need to consider the prior distribution with finite support.

Arbitrarily fix a BIC-BIR mechanism $\cM = (x, p)$.
We consider $\cM$ in its ex ante form.
%and the linear programming of maximizing the broker's profit.
Specifically,
for any buyer $i$ with valuation $v_i$ and for any item $j$, let
(a) $x_{ij}(v_{i}) \triangleq \bE_{v_{-i}\sim D_{-i}}x_{ij}(v_{i},v_{-i})$
be the probability that buyer $i$ gets item $j$,
over the randomness of the mechanism and $v_{-i}$;
and (b) $p_{i}(v_{i}) \triangleq \bE_{v_{-i}\sim D_{-i}}p_{i}(v_{i},v_{-i})$ be the expected payment
made by
 $i$.
%over the randomness of the mechanism and $v_{-i}$.
Then the broker's profit is
$$
PFT(\cM;\cI^{c}) = \sum_{i \in  N} \sum_{v_{i} \in T_{i}} D_{i}(v_{i})  \left( p_{i}(v_{i}) - \sum_{j}x_{ij}(v_{i})c_{j}\right),
$$
where $D_i(v_i)$ is the probability of $v_i$ according to distribution $D_i$.

Consider the constraints for $\cM$.
%To simplify our analysis, we relax the DSIC and EIR constraints to BIC and BIR constraints,
%thus the resulting profit is an upper bound of the optimal DSIC profit.
Let $T^{+}_{i}=T_{i}\cup\{\bot\}$, where ``$\bot$'' is a special symbol not in the support $T_{i}$.
Moreover, $x_{i}(\bot)\triangleq (0,\dots, 0)$ and $p_{i}(\bot)\triangleq 0$.
Because $\cM$ is BIC and BIR, we have
\begin{eqnarray}
x_{i}(v_{i}) \cdot v_{i} - p_{i}(v_{i}) \geq x_{i}(v'_{i}) \cdot v_{i} - p_{i}(v'_{i}), 
\hspace{20pt}
\forall i\in  N, v_{i} \in T_{i}, v'_{i} \in T^{+}_{i} \label{eq:add:ic}
\end{eqnarray}
In particular, when $v'_{i} = \bot$, the above constraint restricts the mechanism to be BIR.
Therefore, the problem of designing BIC-BIR mechanisms is to maximize $PFT(\cM;\cI^{c})$
subject to Inequality \ref{eq:add:ic}.

To upper-bound $PFT(\cM;\cI^{c})$, we first introduce some notations.
For any buyer $i$, item $j$ and valuation sub-profile
$v_{-i} \in T_{-i}$, let $P_{ij}(v_{-i}) = \max_{i'\neq i} v_{i'j}$
and $\beta_{ij}(v_{-i}) = \max\{P_{ij}(v_{-i}), c_{j}\}$,
as in Mechanism $\cM_{BVCG}$.
In Mechanism \ref{alg:1la} we generalize the 1-lookahead mechanism \cite{ronen2001approximating} to production-cost markets, denoted by $\cM_{1LA}$.
For any buyer $i$ with valuation $v_{i}$ and for any item $j$,
let $r_{ij}(v_{-i}) = \max_{p\geq \beta_{ij}(v_{-i})} (p - c_j) \cdot \Pr_{y\sim D_{ij}}[y\geq p]$
and $r_{i}(v_{-i}) = \sum_{j \in  M} r_{ij}(v_{-i})$.
Moreover, let $r_{i} = \bE_{v_{-i}\sim D_{-i}} r_{i}(v_{-i})$ and $r = \sum_{i} r_{i}$.
Then the profit of $\cM_{1LA}$ is exactly $r$
and will be at most $PFT(\cM_{IT}; \cI^{c})$.

\begin{algorithm}[htbp]
\floatname{algorithm}{Mechanism}
  \caption{\hspace{-3pt} $\cM_{1LA}$ for Production-Cost Markets}
 \label{alg:1la}
  \begin{algorithmic}[1]
%\REQUIRE Distribution $D$ and value $v$.
\STATE Collect the valuation profile $v$ from the buyers.

\FOR{each item $j$}
\STATE If no buyer has value for $j$ higher than $\beta_{ij}(v_{-i})$, keep $j$ unsold.

\STATE Otherwise, let $i$ be the highest bidder for $j$, and
 $$\rho_{ij}(v_{-i}) \triangleq \argmax\limits_{p \geq \beta_{ij}(v_{-i})} (p - c_j) \cdot \Pr_{y \sim D_{ij}} (y \geq p).$$
Sell the item $j$ to buyer $i$ with price $\rho_{ij}(v_{-i})$ if and only if $v_{ij} \geq \rho_{ij}(v_{-i})$.
\ENDFOR
\end{algorithmic}
\end{algorithm}

Next, let $R^{(v_{-i})}_{0} = \{v_{i}\in T_{i} | v_{ij} \leq \beta_{ij}(v_{-i}) \mbox{ for all $j$} \}$
and
$R^{(v_{-i})}_{j} = \{v_{i}\in T_{i} | \mbox{$j$ is the smallest index such that }
j\in\arg\max_{k\in M}\{v_{ik} - \beta_{ik}(v_{-i})\} \mbox{ and } v_{ij} - \beta_{ij}(v_{-i})>0\}$.
Moreover, let $\bI[E]$ be the indicator of an event $E$.
Now we are ready to upper-bound $PFT(\cM;\cI^{c})$.
By adopting  the duality framework in \cite{cai2016duality} to production-cost markets,
we have the following.

\begin{lemma}
\label{lem:duality}
For any BIC-BIR mechanism $\cM=(x,p)$ and production-cost instance $\cI^{c}=(N,M,D,c)$,
$$
PFT(\cM;\cI^{c}) \leq
\mbox{\sc Single + Under + Over + Tail + Core},
$$
where $\phi_{ij}(v_{ij})$ is Myerson's (ironed) virtual value  and
%if $D$ is a buyer's valuation distribution, then the virtual value function is $\phi(v)=v-\frac{1-F(v)}{f(v)}$

\begin{eqnarray*}
\mbox{\sc Single}& =&
 \sum_{i \in  N} \sum_{v_{i} \in T_{i}}
\sum_{j \in  M}
D_{i}(v_{i})
x_{ij}(v_{i}) (\phi_{ij}(v_{ij}) - c_{j})
\cdot\Pr_{t_{-i}\sim D_{-i}}
\left[v_{i} \in R^{(t_{-i})}_{j}\right],\\
\mbox{\sc Under} &=&
\sum_{i \in  N} \sum_{v_{i} \in T_{i}}
\sum_{j \in  M}
D_{i}(v_{i})
x_{ij}(v_{i})  
\cdot \left(\sum_{t_{-i}\in T_{-i}} D_{-i}(t_{-i})  (v_{ij} - c_{j}) \cdot \bI[v_{ij} < \beta_{ij}(t_{-i})] \right),
\end{eqnarray*}

\begin{eqnarray*}
 \mbox{\sc Over} & = &
\sum_{i \in  N} \sum_{v_{i} \in T_{i}}
\sum_{j \in  M}
D_{i}(v_{i})
x_{ij}(v_{i}) 
\cdot \big(\sum_{t_{-i}\in T_{-i}} D_{-i}(t_{-i})
(\beta_{ij}(t_{-i}) - c_{j}))
\cdot\bI[(v_{ij} \geq \beta_{ij}(t_{-i})]\big),\\
\mbox{\sc Tail}& = &
\sum_{i \in  N} \sum_{j \in  M}
\sum_{t_{-i}\in T_{-i}} D_{-i}(t_{-i})  
\sum_{v_{ij} > \beta_{ij}(t_{-i}) + r_i(t_{-i})}
D_{ij}(v_{ij})
(v_{ij} - \beta_{ij}(t_{-i})) \\
&&
\quad\cdot \Pr_{v_{i,-j} \sim D_{i,-j}}
[\exists k \neq j, v_{ik} - \beta_{ik}(t_{-i})
\geq v_{ij} - \beta_{ij}(t_{-i})], \\
\mbox{\sc Core}&=&
\sum_{i \in  N} \sum_{j \in  M}
\sum_{t_{-i}\in T_{-i}} D_{-i}(t_{-i})
\cdot\sum_{\beta_{ij}(t_{-i}) \leq v_{ij}
\leq \beta_{ij}(t_{-i}) + r_i(t_{-i})}
D_{ij}(v_{ij}) 
\cdot(v_{ij} - \beta_{ij}(t_{-i})).
\end{eqnarray*}
\end{lemma}

In the following, we use our mechanisms to bound the above terms separately.
For any production-cost instance $\cI^{c}=(N,M,D,c)$,
we define its corresponding single-parameter COPIES instance $\hat{\cI}^{c}=(\hat{N},M,D,c)$.
More precisely, for each buyer $i\in N$, there are $m$ copies for $i$ in set $\hat{N}$,\
and the $j$-th copy of $i$ (also denoted by $(i,j)$) is only interested in item $j$.
Since in $\cI^{c}$, buyers' valuations are additive, $OPT(\hat{\cI}^{c}) = PFT(\cM_{IT}; \cI^{c})$.
%where $OPT_{B}(\hat{\cI}^{c})$ is the optimal profit generated by any BIC mechanism.

In the following, we first use $OPT(\hat{\cI}^{c})$ to bound {\sc Single}, {\sc Under} and {\sc Over},
which implies that these terms are upper-bounded by $PFT(\cM_{IT}; \cI^{c})$.

\begin{lemma}\label{lem:single}
$\mbox{\sc Single} \leq OPT(\hat{\cI}^{c})$.
\end{lemma}
\begin{proof}
Note that if the allocation rule of $\cM$ is feasible for $\cI^{c}$,
then it is also feasible for $\hat{\cI}^{c}$. Thus,

\begin{eqnarray*}
 \mbox{\sc Single}
&\leq& \sum_{i \in  N} \sum_{v_{i} \in T_{i}}
\sum_{j \in  M}
D_{i}(v_{i})
x_{ij}(v_{i}) (\phi_{ij}(v_{ij}) - c_{j}) 
\leq OPT(\hat{\cI}^{c}).
\end{eqnarray*}
The first inequality holds because
$\Pr_{t_{-i}\sim D_{-i}}
[v_{i} \in R^{(t_{-i})}_{j}] \leq 1$,
and the last inequality holds because the optimal mechanism for single-parameter buyers
maximizes the virtual social welfare.
\end{proof}

Next we upper-bound $\mbox{\sc Under}$.
Essentially, the proof of Lemma \ref{lem:under} shares the same idea with Lemma 15 in \cite{cai2016duality},
but here we have to deal with multiple reserve prices for each item.

\begin{lemma}\label{lem:under}
$\mbox{\sc Under} \leq OPT(\hat{\cI}^{c})$.
\end{lemma}
\begin{proof}
It suffices show that
for any BIC-BIR mechanism $\cM$ for $\cI^{c}$,
there is a mechanism $\cM'$ for the single-parameter instance $\hat{\cI}^{c}$,
which achieves profit at least as much as {\sc Under}.

Given $\cM=(x,p)$,
we define $\cM'=(x',p')$ as follows.
Randomly draw $v\sim D$ and run $\cM$ on $v$.
If item $j$ is sold to buyer $i$, let $v_{ij}$ be item $j$'s reserve price in $\cM'$.
Then for instance $\hat{\cI}^{c}$ with reported valuation profile $t$,
run second-price auction with multiple reserve prices for each item.
That is for item $j$, the highest-bid buyer, say $(i_{j},j)$,
gets the opportunity to buy it and the price is $\max\{v_{ij},\beta_{i_{j}j}(t_{-i_{j}})\}$.
If $t_{i_{j}j}<\max\{v_{ij}, \beta_{i_{j}j}(t_{-i_{j}})\}$, item $j$ is kept unsold.

Next we compare the profit of $\cM'$ and {\sc Under}.
%We rearrange {\sc Under} as follows. % is maximized under deterministic allocation rule.
Recall that $x_{ij}(v_{i}) = \sum_{v_{-i}\sim T_{-i}} D_{-i}(v_{-i}) x_{ij}(v_{i},v_{-i})$.
Thus {\sc Under} can be rearranged as follows.

\begin{eqnarray*}
&& \mbox{\sc Under} \\
&=& \sum_{i \in  N} \sum_{v_{i} \in T_{i}}
\sum_{j \in  M}
D_{i}(v_{i})
x_{ij}(v_{i})
\cdot \sum_{t_{-i}\in T_{-i}}
D_{-i}(t_{-i})  (v_{ij} - c_j)
\cdot\bI[v_{ij} < \beta_{ij}(t_{-i})] \\
&=& \sum_{i \in  N} \sum_{v \in T}
\sum_{j \in  M} D(v) x_{ij}(v)
\sum_{t\in T}
D(t) (v_{ij} - c_j)
\cdot\bI[v_{ij} < \beta_{ij}(t_{-i})] \\
&=& \sum_{v \in T}
D(v)(
\sum_{i \in  N} \sum_{t \in T} D(t)
\sum_{j \in  M}x_{ij}(v) (v_{ij} - c_j)
\cdot\bI[v_{ij} < \beta_{ij}(t_{-i})] ).
\end{eqnarray*}
%The inequality holds by Proposition 3 of \cite{cai2016duality}.
It is easy to see that {\sc Under} is maximized by some deterministic allocation.
So we focus on a deterministic allocation rule $x$.

Let us consider the innermost summation.
For any item $j$, denote by $i$ the buyer such that $x_{ij}(v)=1$ and $\bI[v_{ij} < \beta_{ij}(t_{-i})]=1$.
Then $v_{ij}$ is one of the reserve prices for item $j$ in $\cM'$.
Next we see how much profit does mechanism $\cM'$ make from selling item~$j$.

If $c_{j}\geq P_{ij}(t_{-i})$,
$\bI[v_{ij} < \beta_{ij}(t_{-i})]=1$ implies $v_{ij}< c_{j}$.
In $\cM'$, item $j$ will be either sold with price at least $c_{j}$ or unsold,
which means the profit is always nonnegative from selling item $j$ and naturally greater than $v_{ij}-c_{j}<0$.

If $c_{j} < P_{ij}(t_{-i})$,
$\bI[v_{ij} < \beta_{ij}(t_{-i})]=1$ implies $v_{ij}< P_{ij}(t_{-i})$.
Note that the highest bidder for item $j$ (denoted by $(i_{j},j)$) must have value at least $P_{ij}(t_{-i})$.
Since $P_{ij}(t_{-i})>v_{ij}$ and $P_{ij}(t_{-i})>c_{j}$, $(i_{j},j)$ is willing to take item $j$ at price
$\max\{v_{ij},c_{j},P_{i_{j}j}(t_{-i_{j}})\}$ which is at least $v_{ij}$.
Then the profit is at least $v_{ij}-c_{j}$.

%Then there must be some buyer $(i'_{j},j)$ with $t_{i'_{j}j} \geq P_{ij}(t_{-i})>v_{ij}$
%such that $(i'_{j},j)$ is a winner in $\cM'^{c}$ and pays at least $v_{ij}$.
%%such that $t_{i'j} \geq P_{ij}(t_{-i})$

Combining the above two cases,
\begin{eqnarray*}
PFT(\cM';\hat{\cI}^{c}) &\geq & \sum_{i \in  N} \sum_{t \in T} D(t)
\sum_{j \in  M}x_{ij}(v) (v_{ij} - c_j)
\cdot\bI[v_{ij} < \beta_{ij}(t_{-i})].
\end{eqnarray*}

Therefore,
\begin{eqnarray*}
\mbox{\sc Under} \leq \sum_{v \in T} D(v)PFT(\cM';\hat{\cI}^{c}) \leq OPT(\hat{\cI}^{c}),
\end{eqnarray*}
which finishes the proof of Lemma \ref{lem:under}.
\end{proof}

$\mbox{\sc Over}$ is also upper-bounded by $OPT(\hat{\cI}^{c})$,
and the proof is almost the same with Lemma 14 of \cite{cai2016duality},
thus omitted here.

\begin{lemma}\label{lem:over}
$\mbox{\sc Over} \leq OPT(\hat{\cI}^{c})$.
\end{lemma}

Next, we use $PFT(\cM_{1LA};\cI^{c})$ and $PFT(\cM_{BVCG};\cI^{c})$ to bound {\sc Tail} and {\sc Core}.
Recall that $PFT(\cM_{1LA};\cI^{c})=r \leq OPT(\hat{\cI}^{c})=PFT(\cM_{IT}; \cI^{c})$.

\begin{lemma}\label{lem:tail}
$\mbox{\sc Tail} \leq r$.
\end{lemma}

\begin{proof}
\begin{eqnarray*}
&& \mbox{\sc Tail} \\
&\leq&
\sum_{i \in  N} \sum_{j \in  M}
\sum_{t_{-i}\in T_{-i}} D_{-i}(t_{-i})
\sum_{v_{ij} > \beta_{ij}(t_{-i}) + r_i(t_{-i})}
D_{ij}(v_{ij}) 
\cdot (v_{ij} - \beta_{ij}(t_{-i}) + \beta_{ik}(t_{-i}) - c_k) \\
 &&~~~~~
 \quad\quad \cdot \Pr_{v_{i,-j} \sim D^B_{i,-j}}
[\exists k \neq j, v_{ik} - \beta_{ik}(t_{-i})
\geq v_{ij} - \beta_{ij}(t_{-i})] \\
&\leq&
\sum_{i \in  N} \sum_{j \in  M}
\sum_{t_{-i}\in T_{-i}} D_{-i}(t_{-i})
\sum_{v_{ij} > \beta_{ij}(t_{-i}) + r_i(t_{-i})}
D_{ij}(v_{ij}) 
\cdot (v_{ij} - \beta_{ij}(t_{-i}) + \beta_{ik}(t_{-i}) - c_k) \\
&&~~~~~
 \quad\quad  \cdot \sum_{k\neq j} \Pr_{v_{ik} \sim D^B_{ik}}
[v_{ik} - \beta_{ik}(t_{-i})
\geq v_{ij} - \beta_{ij}(t_{-i})] \\
&\leq&
\sum_{i \in  N} \sum_{j \in  M}
\sum_{t_{-i}\in T_{-i}} D_{-i}(t_{-i})
\cdot\sum_{v_{ij} > \beta_{ij}(t_{-i}) + r_i(t_{-i})}
D_{ij}(v_{ij})
\sum_{k\neq j} r_{ik}(t_{-i}) \\
&\leq&
\sum_{i \in  N} \sum_{j \in  M}
\sum_{t_{-i}\in T_{-i}} D_{-i}(t_{-i}) r_{i}(t_{-i})
\cdot\sum_{v_{ij} > \beta_{ij}(t_{-i}) + r_i(t_{-i})} D_{ij}(v_{ij}) \\
&\leq&
\sum_{i \in  N} \sum_{j \in  M}
\sum_{t_{-i}\in T_{-i}} D_{-i}(t_{-i})
(r_{i}(t_{-i}) + \beta_{ij}(t_{-i}) -c_{j})
\cdot \sum_{v_{ij} > \beta_{ij}(t_{-i}) + r_i(t_{-i})}
D_{ij}(v_{ij}) \\
&\leq&
\sum_{i \in  N} \sum_{j \in  M}
\sum_{t_{-i}\in T_{-i}} D_{-i}(t_{-i})
r_{ij}(t_{-i})  = r.\\
\end{eqnarray*}
The first and fifth inequalities are because  $\beta_{ik}(t_{-i}) \geq c_k$ for all $k\in M$.
The second inequality is by union bound.
The third and sixth inequalities hold by the definition of $r_{ik}(t_{-i})$.
\end{proof}

\begin{lemma}\label{lem:core}
$\mbox{\sc Core} \leq 2r + 2PFT(\cM_{BVCG}; \cI^c)$.
\end{lemma}

\begin{proof}
For the random variable $v_{ij}\sim D_{ij}$, define the following two new random variables:
$b_{ij}(t_{-i}) = (v_{ij} - \beta_{ij}(t_{-i}))
\bI [v_{ij} \geq \beta_{ij}(t_{-i})]$
and $d_{ij}(t_{-i}) = b_{ij}(t_{-i})
\bI [b_{ij}(t_{-i}) \leq r_i(t_{-i})]$.
Note that $\sum_j b_{ij}(t_{-i})$ is buyer $i$' utility in the VCG mechanism with reserve price $c$.
%value surplus over item reserve $\beta_{ij}(t_{-i})$.
%$$\mathop{\bE}\limits_{v_i \sim D^B_i} d_{ij}(t_{-i})
%= \sum_{\beta_{ij}(t_{-i}) \leq v_{ij}
%\leq \beta_{ij}(t_{-i}) + r_i(t_{-i})}
%f_{ij}(v_{ij}) (v_{ij} - \beta_{ij}(t_{-i})). $$
Then
$$\mbox{\sc Core} =
\sum_{i \in  N} \sum_{j \in  M}
\sum_{t_{-i}\in T_{-i}} D_{-i}(t_{-i})
\mathop{\bE}\limits_{v_i \sim D^B_i} d_{ij}(t_{-i}).
$$
Set $\hat{e}_i(t_{-i}) = \sum_{j\in  M}\mathop{\bE}\limits_{v_i \sim D^B_i} [d_{ij}(t_{-i})]
- 2 r_i(t_{-i})$.
By Lemma 12 in~\cite{cai2016duality},  %we have
$$
\Pr[\sum_{j \in  M} d_{ij}(t_{-i}) \leq \hat{e}_i(t_{-i})]
\leq \frac{1}{2}.
$$
Thus % we have
$$
\Pr[\sum_{j \in  M} b_{ij}(t_{-i}) \geq \hat{e}_i(t_{-i})]
\geq \Pr[\sum_{j \in  M} d_{ij}(t_{-i}) \geq \hat{e}_i(t_{-i})]
\geq \frac{1}{2}.
$$

Note that the entry fee $e_i(v_{-i})$ is the median of the random variable $\sum_{j\in  M}(v_{ij} - \beta_{ij}(v_{-i}))^{+}$,
then $e_i(v_{-i}) \geq \hat{e}_i(v_{-i})$ and
$\Pr[\sum_{j \in  M} b_{ij}(t_{-i}) \geq e_i(t_{-i})] = \frac{1}{2}$.
Therefore,
\begin{equation*}
PFT(\cM_{BVCG}; \cI^c) \geq \frac{1}{2}\sum_{i \in  N}
\sum_{t_{-i}\in T_{-i}} D_{-i}(t_{-i}) \hat{e}_i(t_{-i}) 
= \frac{\mbox{\sc Core}}{2} - r. \qedhere
\end{equation*}
\end{proof}

%mechanism $\cM_A$ is cost-monotone by combining Lemma \ref{lem:singlemono} and \ref{lem:additive}.
Finally, we obtain the main theorem of this section.
\setcounter{theorem}{1}
\begin{theorem}
\label{thm:cost}
When the buyers have additive valuations,
Mechanism $\cM_{A}$ is DSIC and is an 8-approximation to the optimal profit for production-cost markets.
\end{theorem}

\begin{proof}
Combining Lemma \ref{lem:duality} with Lemmas \ref{lem:single}, \ref{lem:under}, \ref{lem:over}, \ref{lem:tail} and \ref{lem:core},
for any BIC mechanism $\cM=(x,p)$ for production-cost markets and any production-cost market instance $\cI^{c}=(N,M,D,c)$,
$$
PFT(\cM;\cI^{c}) \leq 6 PFT(\cM_{IT};\cI^{c})+2PFT(\cM_{BVCG}; \cI^c).
$$
By selling items using $\cM_{IT}$  with probability $\frac{3}{4}$ and using
$\cM_{BVCG}$ with probability $\frac{1}{4}$, we have an 8-approximation to the optimal profit.
\end{proof}

\end{document}